\documentclass[letterpaper, 10 pt, conference]{ieeeconf}  

\IEEEoverridecommandlockouts                              
\usepackage{cite}
\usepackage{amsmath,amssymb,mathrsfs,mathtools}
\usepackage{algorithmic}
\usepackage{graphicx,subfigure}
\usepackage{xcolor}
\usepackage{hyperref}
\hypersetup{hidelinks=true}

\usepackage{enumerate}
\usepackage{color}
\usepackage{booktabs}

\allowdisplaybreaks[4]

\newtheorem{asmp}{\bf\emph{Assumption}}

\newtheorem{prop}{\bf\emph{Proposition}}
\newtheorem{thm}{\bf\emph{Theorem}}

\newtheorem{rmk}{\bf\emph{Remark}}
\newtheorem{defn}{\bf\emph{Definition}}

\newcommand{\RR}{\mathbb{R}}

\newcommand{\II}{\mathbb{I}}

\newcommand{\PP}{\mathbb{P}}
\newcommand{\NN}{\mathbb{N}}

\newcommand{\eps}{\varepsilon}

\newcommand{\mtcb}{\mathcal{B}}

\newcommand{\mtce}{\mathcal{E}}

\newcommand{\mtcg}{\mathcal{G}}

\newcommand{\mtci}{\mathcal{I}}

\newcommand{\mtcs}{\mathcal{S}}

\newcommand{\mtcv}{\mathcal{V}}

\newcommand{\mtcx}{\mathcal{X}}

\newcommand{\Let}{\coloneqq}
\newcommand{\teL}{\eqqcolon}
\newcommand{\sgn}{\textup{sgn}}

\newcommand{\bfb}{\mathbf{b}}

\newcommand{\bfo}{\mathbf{0}}
\newcommand{\bfv}{\mathbf{v}}
\newcommand{\bfx}{\mathbf{x}}
\newcommand{\bfS}{\mathbf{S}}

\newcommand{\tx}{\textup}
\newcommand{\tp}{\top}

\newcommand{\BN}{\mtcb^{(N)}} 
\newcommand{\GW}{W} 
\newcommand{\WN}{W^{(N)}} 
\newcommand{\barAN}{\bar{A}^{(N)}} 
\newcommand{\baraN}{\bar{a}^{(N)}}
\newcommand{\TW}{T_{\GW}}
\newcommand{\TWN}{T_{\WN}} 

\newcommand{\phis}{\varphi^{*}}

\newcommand{\bfxN}{\bfx^{(N)}}
\newcommand{\xN}{x^{(N)}}
\newcommand{\AN}{A^{(N)}}
\newcommand{\aN}{a^{(N)}}
\newcommand{\bfbN}{\bfb^{(N)}}
\newcommand{\bN}{b^{(N)}}
\newcommand{\thetaN}{\theta^{(N)}}
\newcommand{\hatWN}{\hat{W}^{(N)}}
\newcommand{\ThWN}{T_{\hatWN}}
\newcommand{\xiN}{\xi^{(N)}}
\newcommand{\phiN}{\varphi^{(N)}}

\newcommand{\bfvN}{\bfv^{(N)}}

\newcommand{\zetaN}{\zeta^{(N)}}
\newcommand{\zetas}{\zeta^{*}}

\title{\LARGE \bf
Behavior of Nonlinear Opinion Dynamics over Large Networks
}

\author{Yu Xing, Anastasia Bizyaeva, and Karl H. Johansson
\thanks{This work was supported by the Alexander von Humboldt Foundation, Knut and Alice Wallenberg Foundation (Wallenberg Scholar grant), the Swedish Research Council (Distinguished Professor grant 2017-01078), and the Swedish Foundation for Strategic Research (SUCCESS FUS21-0026).}
\thanks{YX is with the Faculty of Computer Science, RWTH Aachen University, Aachen, Germany. 
Email: {\tt\small yu.xing@rwth-aachen.de}. 
AB is with the Sibley School of Mechanical and Aerospace Engineering, Cornell University, Ithaca, NY, USA. Email: {\tt\small anastasiab@cornell.edu}.  
KHJ is with Division of Decision and Control Systems, School of Electrical Engineering and Computer Science, KTH Royal Institute of Technology, and also with Digital Futures, Stockholm, Sweden.
Email: {\tt\small {kallej}@kth.se}.}%
}

\begin{document}

\maketitle
\thispagestyle{empty}
\pagestyle{empty}

\begin{abstract}
    Opinion dynamics have been studied for decades across disciplines, with much of the theoretical literature focusing on behaviors such as consensus, polarization, and clustering.
    Although classic models can exhibit more complex opinion patterns in simulations, quantifying such distributions is not fully understood.
    To address this question, in this paper, we study the behavior of nonlinear opinion dynamics over large-scale networks using graphons, which capture the underlying network structure. 
    In the model, agents update their opinions according to a nonlinear rule that includes saturation effects in interactions. 
    The network is represented by random graphs generated from a graphon, and a corresponding nonlinear dynamical model is introduced over the graphon. 
    We show that the graphon dynamics approximate the finite-dimensional system, when the network size is large. 
    Leveraging spectral approximation results for random graphs, we further show that the equilibria of the nonlinear model can also be approximated by those of the continuum limit. 
    This result enables a quantitative characterization of the opinion distribution based on the underlying graphon structure.
    The theoretical results are illustrated by numerical simulation.
\end{abstract}

\section{INTRODUCTION}
Opinion dynamics have drawn attention for decades across various disciplines, including physics, mathematics, and control systems, due to their wide presence and potential applications in economics, politics, and management~\cite{castellano2009statistical,flache2017models}.
Most traditional studies focus on deriving mathematical conditions under which models converge to typical behaviors such as consensus, polarization, and clustering~\cite{proskurnikov2017tutorial}.
However, real-world opinion formation often exhibits more diverse patterns, for instance, continuous distributions with multiple peaks~\cite{flache2017models,devia2022framework}.
Although classic models are able to reproduce such patterns in simulations~\cite{flache2017models}, the mechanisms governing their emergence are not yet fully characterized theoretically. 

As a first step towards addressing this gap, recent studies~\cite{xing2024concentration,xing2024transient,wang2024final} have characterized how network structure influences final opinion distributions for gossip and the Friedkin--Johnsen models.
Despite the linear evolution, these models can generate various distributions, for instance, with peaks corresponding to network clusters. 
However, real opinion distributions typically do not admit such direct correspondence to network structure, which raises the question of how network structure affects opinion patterns in nonlinear dynamics.
Addressing this question can enable the approximation of large-scale opinion systems by low-dimensional network models, allowing us to predict opinion transition between different collective behaviors (e.g., from consensus to disagreement) with quantitative guarantees, and facilitate learning~\cite{xing2024learning} and control tasks~\cite{kohler2025integer}.

\subsection{Related Work}
This paper considers real-valued opinion dynamics, which provide a theoretical framework for modeling opinion evolution. 
The most basic model is the French--DeGroot model~\cite{degroot1974reaching}, where agents update their opinions by averaging those of their neighbors, leading to group consensus.
Extensions include the Friedkin--Johnsen model~\cite{friedkin1990social}, bounded confidence models~\cite{deffuant2000mixing}, and dynamics with antagonistic interactions~\cite{altafini2012consensus}. 
Unlike these models, nonlinear variants~\cite{fontan2021signed,bizyaeva2022nonlinear,baumann2020modeling} incorporate saturation-based update mechanisms, similar to those observed in biological systems~\cite{franci2015realization}.
These models can capture opinion transitions from consensus to disagreement even in the absence of bias or inputs, and have been used to explain phenomena such as political polarization~\cite{leonard2021nonlinear}, echo chambers on social media~\cite{baumann2020modeling}, and multiparty parliament negotiations~\cite{fontan2021signed}.

Real opinion dynamics often involve a large number of individuals interacting over complex networks.
Such networks often exhibit patterns such as community structure, which can be captured by low-complexity models.
This idea can be formalized through graphons, which arise as limits of graph sequences~\cite{lovasz2012large}. The graphon framework has been extended to characterize quantities such as centrality measures~\cite{avella2018centrality} and the Laplacian spectrum~\cite{vizuete2021laplacian}.
Large dynamical systems can also be analyzed via their continuum limits over graphons~\cite{wiley2006size}. 
This approach has been rigorously developed for systems including nonlinear heat dynamics\cite{medvedev2014nonlineara} and stochastic interacting particle systems~\cite{bayraktar2022stationarity}.
Consensus and bipartite consensus have been established in averaging-based opinion dynamics~\cite{qiao2025consensus,prisant2025asymptotic}, and synchronization has been studied in graphon oscillator models~\cite{nagpal2024synchronization}.
However, existing research focuses on convergence behavior and approximation properties, and does not address how to quantify the influence of network structure on the formation of opinion distributions, which is the focus of this paper.

\subsection{Contributions}
We study nonlinear opinion dynamics through a continuum-limit approximation based on a graphon.
The nonlinear opinion dynamics~\cite{bizyaeva2021patterns,bizyaeva2022nonlinear} describe how agents' attention to nonlinear network interactions can result in a range of asymptotic behaviors, from consensus to disagreement, even without bias or external inputs.

By representing the underlying network structure with a graphon, we establish high-probability bounds on the state deviation between the original finite-dimensional dynamics and its continuum counterpart (Theorem~\ref{thm:dynamic_approx}).
It is shown that the evolution of the finite-dimensional system with a fixed network size can be predicted by the continuum model over a finite time horizon.
Moreover, this horizon increases with the network size.
This result demonstrates how the influence of network structure on system evolution can be reflected by the continuum model over the graphon.

We further show that, when the attention exceeds a certain threshold, the emergent equilibria of the finite-dimensional system can be approximated by the leading eigenfunction of the graphon (Theorem~\ref{thm:vector}).
This result quantifies the approximation error between the opinion distribution and its continuum limit.
To illustrate the consequence of this approximation, we show that the distribution function of agent opinions converges to that of the continuum system, as the network size tends to infinity, under a vanishing $L^{2}$-norm discrepancy between the two opinion functions (Theorem~\ref{thm:distribution}).
This result indicates that the distribution of final opinions in the finite-dimensional system is strongly governed by the limit network structure encoded by the graphon.

Such approximation provides a framework for analyzing large-scale systems and determining pattern transitions without directly computing quantities related to high-dimensional networks.
Furthermore, since the distribution of opinions captures global properties, approximating it implies the approximation of all finite moments of the opinions, including commonly used statistics such as the mean and variance.
By viewing observed opinions as samples from an underlying distribution, the developed framework also suggests potential applications in data-driven modeling and learning.

\subsection{Outline}
The paper is organized as follows. 
In Section~\ref{sec:prel}, we introduce nonlinear opinion dynamics and graphon models, and formulate the problem.
Section~\ref{sec:results} provides main theoretical results, and Section~\ref{sec:simulation} presents associated numerical experiments. 
Section~\ref{sec:conclusion} concludes the paper.

\subsubsection*{\bf Notation}
A vector is denoted by a boldface letter, e.g., $\bfx$, and its $i$-th entry by $x_i$. 
For a matrix $A \in \mathbb{R}^{n\times n}$, $a_{ij}$ or $[A]_{ij}$ denotes its $(i,j)$-th entry.
For a symmetric $A \in \RR^{n\times n}$, denote its largest eigenvalue by $\lambda_{\max}(A)$. 
By $I_n$ we denote the identity matrix. 
The unit interval is represented by $\mtci \Let [0,1]$.
The notation $\|\cdot\|_{2}$ is used for the Euclidean norm of a vector, the spectral norm of a matrix or an operator, and the $L^{2}$-norm of a function. 
The function $\II_{[\textup{property}]}$ is the indicator function, which is one if the property in the bracket holds, and zero otherwise, and we write $\II_{\mtcb}(x) \Let \II_{[x\in B]}$ for $x \in \mtcx$, sets $\mtcb$ and $\mtcx$ with $\mtcb \subset \mtcx$. 
An undirected graph is denoted by $\mtcg = (\mtcv,\mtce,A)$, where $\mtcv$ is the agent set, $\mtce$ is the edge set, and $A = [a_{ij}]$ is the adjacency matrix such that $a_{ij} = 1$ ($a_{ij} = 0$) if $\{i,j\}\in \mtce$ ($\{i,j\}\not\in\mtce$).

\section{Preliminaries and Problem Statement}\label{sec:prel}
In the first two subsections, we introduce nonlinear opinion dynamics and graphon models.
The problem studied in this paper is then formulated in Section~\ref{subsec:formulation}.

\subsection{Nonlinear Opinion Dynamics}
We consider the following nonlinear opinion dynamics taking place over an undirected graph $\mtcg = (\mtcv,\mtce, \AN)$ with $\mtcv = \{1,\dots,N\}$ and no self-loops ($a_{ii}=0$). 
Each agent $i \in \mtcv$ has a state $\xN_{i}(t) \in \RR$, at continuous time $t\in \RR_+$, and updates the opinion according to
\begin{multline}\label{eq_opinion_model_agentform}
    \dot{x}^{(N)}_i = - \bar{d} \xN_i + u S\Big(\alpha \xN_i \\ + \frac{\gamma}{N \kappa_{N}} \sum_{k\in\mtcv} \aN_{ik} \xN_k\Big) + \bN_{i}, 
\end{multline}
where $\bar{d} > 0$ is the damping factor, $u > 0$ is the agent attention parameter to the nonlinear network interaction, and $S= \tanh$, i.e., the hyperbolic tangent. 
Note that $S$ can be other odd sigmoid functions such that $S(0)=0$, $S^{\prime}(0)=1$, and $\sgn(S^{\prime\prime}(z)) = - \sgn(z)$. 
In the nonlinear function $S(\cdot)$, $\alpha \ge 0$ is the self weight, $\gamma > 0$ is the influence weight of other agents, and $\kappa_{N} \in (0,1]$ is a scaling factor depending only on $N$.
The case with $\gamma < 0$ describes the scenario where agents have antagonistic interactions, which will be explored in future work. 
Finally, $\bfbN = [\bN_{i}]_{1 \leq i \leq N}$ captures individual bias or external input. 

The compact form of~\eqref{eq_opinion_model_agentform} can be written as
\begin{align}\label{eq:opinion_model_compactform}
    \dot{\bfx}^{(N)} = - \bar{d} \bfxN + u \bfS \Big(\alpha \bfxN + \frac{\gamma}{N \kappa_{N}} \AN \bfxN \Big) + \bfbN, 
\end{align}
where $\bfS(\bfx) \Let [S(x_1),\dots,S(x_{N})]^\tp$ for $\bfx \in \RR^{N}$ is the entrywise hyperbolic tangent function.
The model and its extensions have been thoroughly studied, focusing on their bifurcation and steady-state behavior~\cite{bizyaeva2022nonlinear,bizyaeva2021patterns}. 
When the attention $u$ is small, network interactions have little influence on agent opinion formation, and all agents reach a consensus at the origin.
As agents increase their attention to opinion exchange over the network beyond a specific threshold, nonlinear dynamics make them adopt an opinion (negative or positive), even without bias or input.
This emergence of new equilibria is characterized by a pitchfork bifurcation, demonstrated by the following result and also in Fig.~\ref{fig:illus_ode}.
\begin{figure*}
    \centering
    \subfigure[\label{fig:illus_consensus}Consensus with a small $u$.]{
        \includegraphics[width=0.26\textwidth]{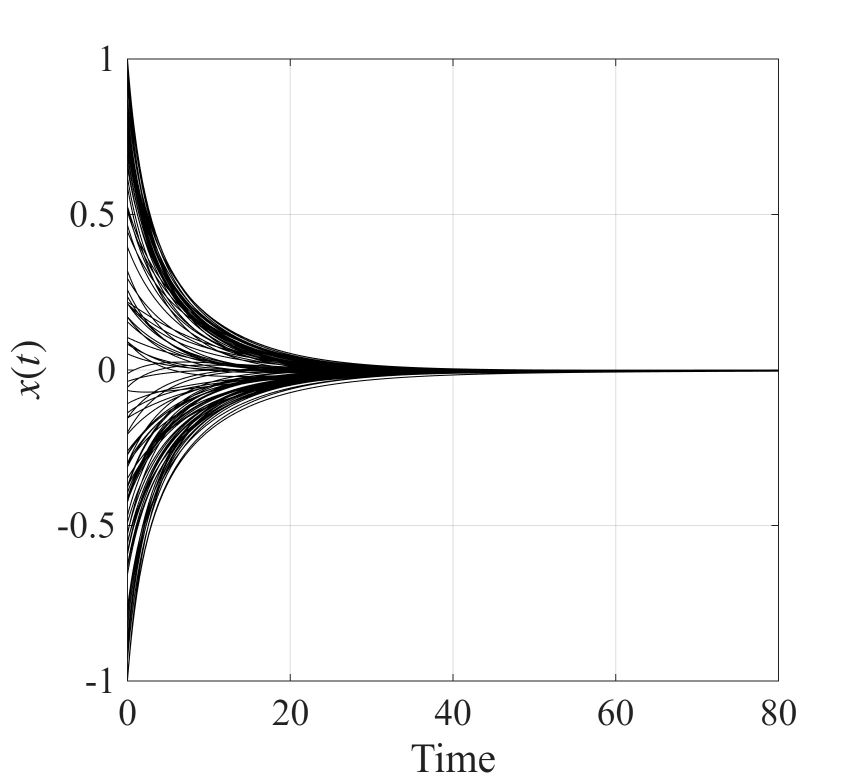} 
    } ~~
    \subfigure[\label{fig:illus_diverse}A new equilibrium with a larger $u$.]{ 
        \includegraphics[width=0.26\textwidth]{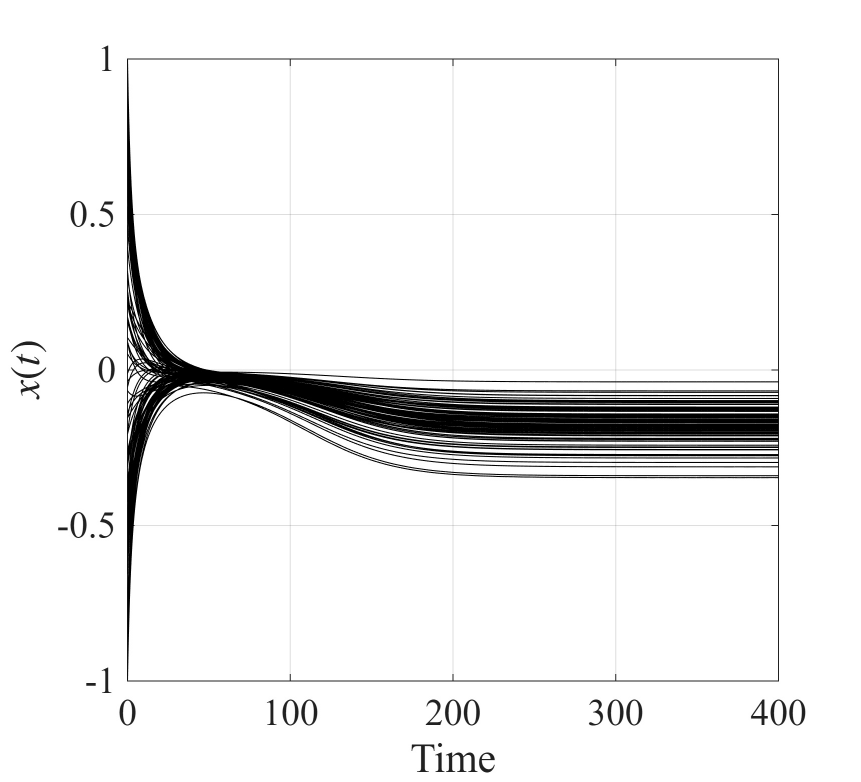}  
    } ~~
    \subfigure[\label{fig:illus_bifur}Pitchfork bifurcation. Blue (red) lines represent stable (unstable) equilibria.]{
        ~~~~~~~
        \includegraphics[width=0.24\textwidth]{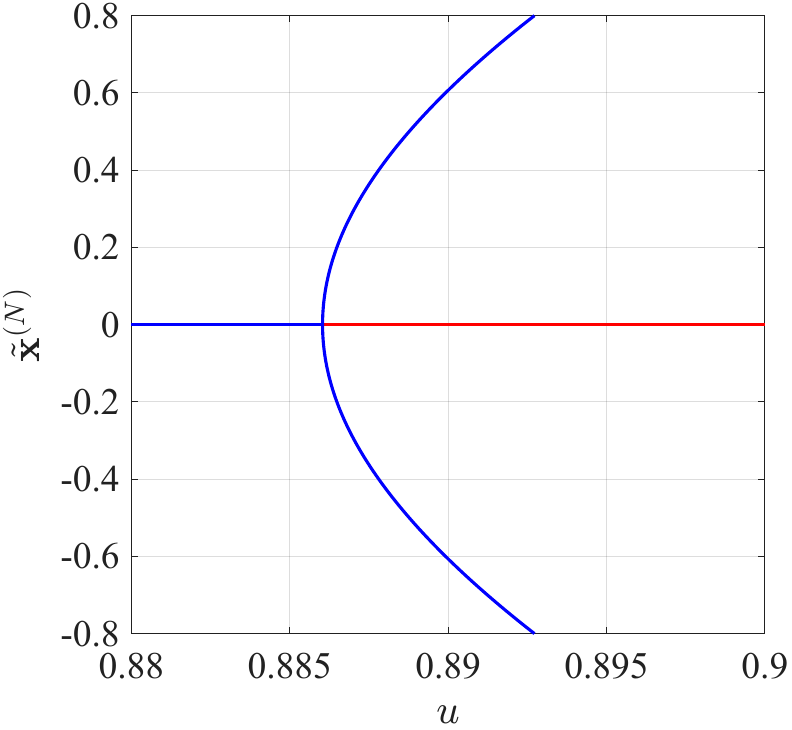}
        ~~~~~~~
    }
    \caption{\label{fig:illus_ode}Illustration of the dynamics and bifurcation of system~\eqref{eq:opinion_model_compactform}.}
\end{figure*}
%

\begin{prop}[Pitchfork bifurcation~\cite{bizyaeva2021patterns}]\label{prop_bifurcation}~\\\indent
    Suppose that $\mtcg$ is connected, $\bar{d}, u, \alpha, \gamma > 0$, and $\bfbN = \bfo$. 
    Then the origin $\bfxN = \bfo$ is a locally exponentially stable equilibrium for $0 < u < u^{*}$ and unstable for $u > u^{*}$, where $u^{*} \Let \bar{d}/(\alpha + \frac{\gamma}{N \kappa_{N}} \lambda_{\max}(\AN))$. 
    At $u = u^{*}$, branches of equilibria $\tilde{\bfx}^{(N)} \not= \bfo$ emerge in a steady-state bifurcation off of $\bfxN = \bfo$ along a manifold tangent to the eigenvector corresponding to $\lambda_{\max}(\AN)$ at $(\bfo, u^{*})$, and the entries of $\tilde{\bfx}^{(N)}$ have the same sign for $|u - u^{*}|$ sufficiently small.
 \end{prop}

The opinion $\xN_i$ shows the agent $i$'s belief about two options. 
The sign $\sgn(\xN_i)$ indicates the option that the agent supports, and $\xN_i = 0$ is the neutral opinion. 
Proposition~\ref{prop_bifurcation} implies that agreement equilibria with an identical sign emerge from the neutral opinion as $u$ increases, in the case of a positive influence weight. 

\subsection{Graphons and Random Graphs}
Graphons provide a quantitative description of network structure.
A graphon is a bounded symmetric measurable function $\GW \colon \mtci^2 \to \mtci$, which can be considered as the limit of a convergent sequence of graphs with increasing size~\cite{lovasz2012large}. 
Each graphon $\GW$ defines an integral linear operator $\TW\colon L^2(\mtci) \to L^2(\mtci)$, which is called the associated graphon operator, where $L^2(\mtci)$ is the $L^{2}$-space on the unit interval $\mtci$. It holds that, for $f \in L^2(\mtci)$,
\begin{align*}
    (\TW f) (y) = \int_{\mtci} \GW(y,z) f(z) dz, ~y \in \mtci.
\end{align*}
This operator can be considered as an extension of the adjacency matrix to the continuum case.

We model a network as a sample drawn from a random graph defined by a graphon as follows.
\begin{defn}[Random graph model]\label{defn:random_graph}
    Given a graphon $\GW$, network size $N\in \NN_+$, and a scaling factor $\kappa_{N} \in (0,1]$, the random graph model $\tx{RG}(\GW, N, \kappa_{N})$ constructs a simple graph $\mtcg = (\mtcv, \mtce, \AN)$ with $\mtcv = \{ 1, \dots, N\}$ according to the following procedure:
    \begin{enumerate}
        \item Construct the weight matrix $\barAN = [\baraN_{ij}] \in [0,1]^{N\times N}$, whose entries are the averages of the graphon $\GW$ over the rectangles $\BN_{i} \times \BN_{j}$
        \begin{align*}
            \baraN_{ij} \Let N^2 \int_{\BN_{i} \times \BN_{j}} W(y,z) dy dz, 
        \end{align*}
        where $\BN_{i} \Let ((i-1)/N, i/N]$ for $1 \leq i \leq N$.        
        \item Define the probability matrix $P^{(N)} = [p^{(N)}_{ij}] \in [0,1]^{N\times N}$ by 
        \begin{align*}
            p^{(N)}_{ij} = 
            \begin{cases}
                \kappa_{N} \baraN_{ij}, &  i \neq j,\\
                0, & i = j,
            \end{cases}
        \end{align*}
        where the scaling factor $\kappa_{N}$ controls the expected average degree of an agent. 
        \item Add the edge $\{i,j\}$ to $\mtce$ with probability $p^{(N)}_{ij}$ independently of other edges.
    \end{enumerate}
\end{defn}

Intuitively, a graphon represents an idealized network structure, and an observed network can be viewed as a noisy realization of this structure.
Fig.~\ref{fig:illus_graphs} provides an illustration, where the graphon has two communities (with agents in $[0, 0.9]$ and $(0.9, 1]$, respectively) with an identical link value within each community, and as a result, the finite network is generated from a stochastic block model.
There are other ways to define a weight matrix $\barAN$ from a graphon, for example, by sampling $N$ points from the unit interval, $x_{1}, \dots, x_{N} \in  \mtci$, and setting $\baraN_{ij} = \GW(x_{i}, x_{j})$~\cite{avella2018centrality}.
The analysis techniques are similar to those used in the current paper.

We introduce the following assumption of piecewise Lipschitz positive graphons.
The piecewise condition is motivated by the community structure shown in Fig.~\ref{fig:illus_graphs}, whereas the positivity ensures the connectivity of the graphon.
\begin{asmp}\label{asmp:graphon_piecewise+positive}~
    \begin{enumerate}
        \item Let $\tx{Lip}_{W}>0$ and $k_{W} \in \NN_+$, and $0 = \iota_{0} < \cdots < \iota_{k_{W}}=1$ be constants, and denote $\mtci_{k} \Let (\iota_{k-1}, \iota_{k}]$, $1 \leq k \leq k_{W}$. 
        For all $1 \leq k, q \leq k_{W}$ and point pairs $(y,z), (y^{\prime}, z^{\prime}) \in \mtci_{k} \times \mtci_{q}$, it holds that
        \begin{align*}
            |\GW(y,z) - \GW(y^{\prime},z^{\prime})| \leq \tx{Lip}_{W}(|y - y^{\prime}| + |z-z^{\prime}|).
        \end{align*}
        \item There exists $\eps_{W} > 0$ such that $\GW(y,z) \geq \eps_{W}$ for all $y,z\in\mtci$.
    \end{enumerate}
\end{asmp}
\begin{figure*}
    \centering
    \subfigure[\label{fig:illus_graphon}The value of the graphon $W$ with two blocks.]{
        ~~~~~~\includegraphics[width=0.28\textwidth]{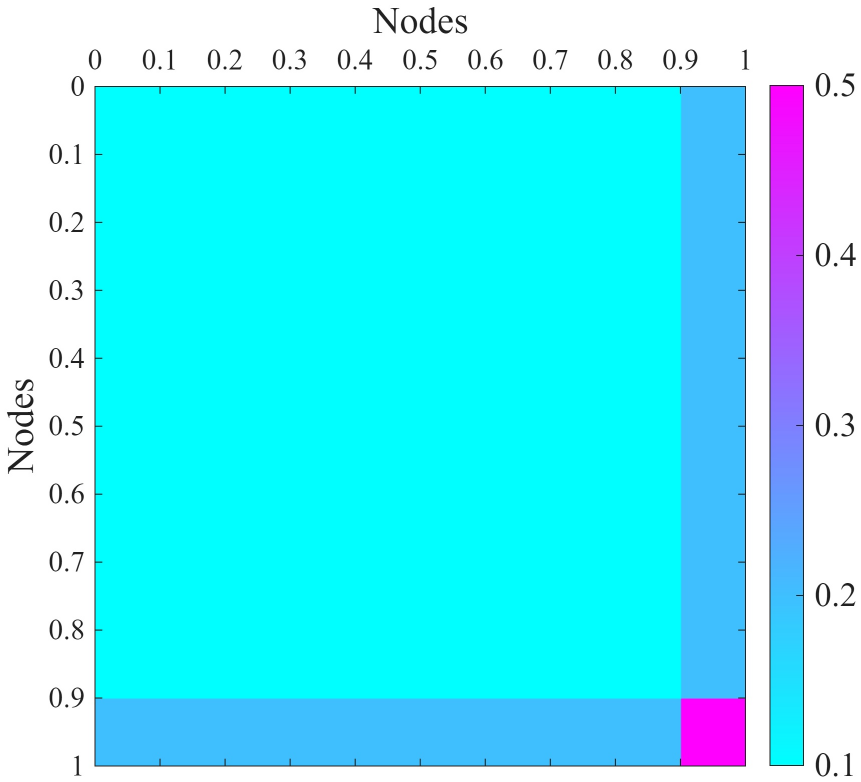} ~~~~~~
    }~~~~~~
    \subfigure[\label{fig:illus_sbm}The network sampled from the random graph model.]{ 
        ~~~~~~\includegraphics[width=0.35\textwidth]{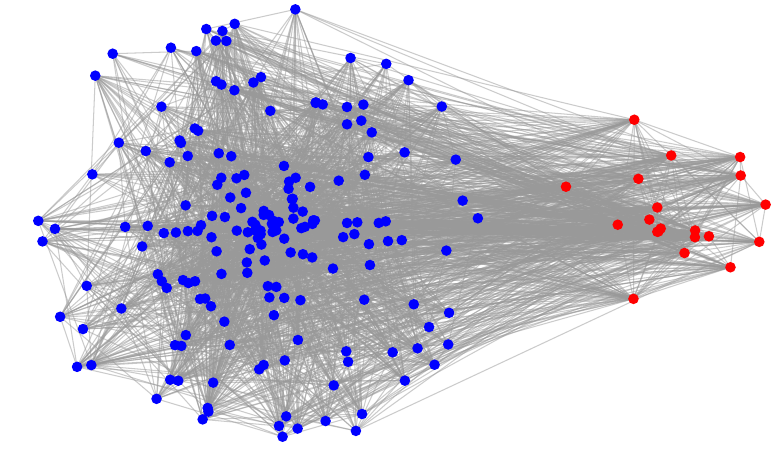}  ~~~~~~
    }
    \caption{\label{fig:illus_graphs}Illustration of a graphon $W$ and a random graph with $N = 200$, where $W(y,z) = 0.1$ for $0 < y, z \leq 0.9$, $W(y, z) = 0.5$ for $0.9 < y, z \leq 1$, and $W(y, z) = 0.2$ otherwise.
    Both the graphon and the random graph have two communities.}
\end{figure*}

\subsection{Problem Formulation}\label{subsec:formulation}
Recall the nonlinear opinion dynamics~\eqref{eq:opinion_model_compactform} over a graph $\mtcg$. 
Suppose that the graph structure can be captured by an idealized model, namely a graphon $\GW$, and we observe the sampled finite network $\mtcg$.
Our goal is to approximate the behavior of the finite-dimensional system using this graphon representation.
To this end, we compare the original model with the following continuum nonlinear dynamics defined over the graphon $\GW$,
\begin{multline}\label{eq:opinion_continuum}
    \frac{\partial \theta(y,t)}{\partial t} = - \bar{d} \theta(y,t) \\
    + u S \bigg ( \alpha \theta(y,t) +  \gamma  \int_{\mtci} \GW(y, z) \theta(z, t) d z \bigg ) + b(y),
\end{multline}
where $y\in \mtci$, $\theta(y,t)$ is the opinion of $y$ at time $t$, and $b(y)$ is the input for agent $y$. 
Note that under regular conditions of the initial condition $\theta(\cdot, 0)$ and input $b$, the existence and uniqueness of the solution to~\eqref{eq:opinion_continuum} can be established, similar to~\cite{medvedev2014nonlineara,nagpal2024synchronization,prisant2025asymptotic}.

We are interested in understanding how well the continuum model captures both the dynamical and asymptotic behaviors of the finite-dimensional system.
In particular, this paper aims to quantify the approximation quality of the continuum model over the graphon, as follows.

{\bf Problem.} Given a graphon $\GW$ and a graph $\mtcg$ generated by $\GW$. Bound the deviation between the states of the finite-dimensional dynamics~\eqref{eq:opinion_model_compactform} over $\mtcg$ and the continuum dynamics~\eqref{eq:opinion_continuum} over $\GW$, and approximate the equilibria of~\eqref{eq:opinion_model_compactform} by using those of~\eqref{eq:opinion_continuum}.

The approximation of dynamics is addressed in Theorem~\ref{thm:dynamic_approx}, which indicates that the evolution of the finite-dimensional system can be predicted by the continuum system over a time horizon increasing with the network size.
The approximation of equilibria is studied in Theorems~\ref{thm:vector} and~\ref{thm:distribution}.
The results describe how the final opinions in the finite-dimensional system and their distributions are governed by the limit network structure given by the graphon. 

\section{Main Results}\label{sec:results}
In this section, we study the approximation of the finite-dimensional dynamics~\eqref{eq:opinion_model_compactform} by its continuum limit~\eqref{eq:opinion_continuum}.
In Section~\ref{sec:approx_dynamics}, we show that the two dynamics are close to each other over a finite time horizon.
In Sections~\ref{sec:approx_equilibria} and~\ref{sec:approx_distribution}, it is shown that the final opinion vector in the finite system can be approximated by its limit, so does the distribution function of the opinions.

\subsection{Approximation of Dynamics}\label{sec:approx_dynamics}
We compare the dynamics of the finite-dimensional system~\eqref{eq:opinion_model_compactform} with that of the continuum model~\eqref{eq:opinion_continuum}.
Note that the former system includes $N$ agents, whereas the latter is characterized by a continuum of agents defined on the unit interval.
As a basis of comparison, consider the piecewise constant interpolant of the system~\eqref{eq:opinion_model_compactform} with dimension $N$ as follows, 
\begin{align}\label{eq:piecewise_interpolant}
    \thetaN(y, t) = \sum_{i = 1}^{N} \xN_{i}(t) \II_{\BN_{i}}(y), ~y \in \mtci,
\end{align}
where $\BN_{i} = ((i - 1) / N, i / N]$, $1 \leq i \leq N$.
Denote the step graphon $\hatWN$ defined by the adjacency matrix $\AN = [\aN_{ij}]$ as follows
\begin{align*}
    \hatWN(y, z) \Let \sum_{k = 1}^{N} \sum_{\ell = 1}^{N} \frac{\aN_{k \ell}}{\kappa_{N}} \II_{\BN_{k}}(y) \II_{\BN_{\ell}}(z), \quad y, z \in \mtci.
\end{align*}
Then system~\eqref{eq:piecewise_interpolant} can be written as a piecewise constant continuum system
\begin{multline}\label{eq:piecewise_interpolant_compact}
    \frac{\partial \thetaN(y, t)}{\partial t} = - \bar{d} \thetaN(y, t) + u S \bigg( \alpha \thetaN(y, t) \\ + \gamma \int_{\mtci} \hatWN(y, z) \thetaN(z, t) dz \bigg) + \bN(y), ~y \in \mtci.
\end{multline}
Here we assume that the initial condition $\thetaN(\cdot, 0)$ and bias $\bN$ of system~\eqref{eq:piecewise_interpolant_compact} are close to those of the continuum system~\eqref{eq:opinion_continuum}, to obtain a vanishing error bound.
In general, the deviation of the finite-dimensional system from the continuum model depends on the discrepancy of both the initial condition and the input.
Specifically, assume that for all $i \in \{1, \dots, N\}$ and $y \in \BN_{i}$, $\xN_{i}(0) = \thetaN(y, 0) = N \int_{\BN_{i}} \theta(y, 0) dy$, and $\bN(y) = N \int_{\BN_{i}} b(y) dy$.
That is, the initial value (input) of an agent~$i$ is given by the average of the continuum initial condition $\theta(\cdot, 0)$ (input $b$) over the interval $\BN_{i}$.

Let $\xiN(y, t) \Let \thetaN(y, t) - \theta(y, t)$ be the difference between the interpolant and the continuum systems. 
We have the following result for the approximation error between the dynamics~\eqref{eq:piecewise_interpolant_compact} and~\eqref{eq:opinion_continuum}.
\begin{thm}[Approximation of dynamics]\label{thm:dynamic_approx}~\\\indent
    Suppose that Assumption~\ref{asmp:graphon_piecewise+positive} holds, and that $\theta(\cdot, 0)$ and $b$ are piecewise Lipschitz continuous. 
    In addition, $\kappa_{N} \geq \log N / (\eps_{W} N)$.
    Then there exist positive constants $r$ and $C$ such that, for any $\eps > 0$ and $T \leq C \log (\eps^{2} N \kappa_{N})$,
    \begin{align*}
        \PP\bigg\{ \sup_{t \in [0, T]} \|\xiN(\cdot, t)\|_{2} \leq \eps \bigg\} \geq 1 - N^{-r},
    \end{align*}
    where $C$ depends on $c_{0}$, $r$, and the parameters $\bar{d}$, $u$, $\alpha$, and $\gamma$ of system~\eqref{eq:opinion_model_compactform}.
\end{thm}
\begin{proof}
    See Appendix~\ref{sec:proof_thm:dynamic_approx}.
\end{proof}

The theorem indicates that the norm of the error $\xiN(\cdot, t)$ can be upper bounded over a finite time interval with high probability.
As the network size grows, the length of the interval increases and the failure probability decreases.
As a consequence, the dynamics of the finite-dimensional model can be predicted by the continuum dynamics over the graphon that quantifies the network structure.
\begin{rmk}
    The probability in the bound comes from the sampling of the random graph.
    To ensure a valid upper bound for the time interval, the expected degree $N \kappa_{N}$ should be much larger than $1 / \varepsilon^{2}$.
    In addition, if $N \kappa_{N} = \Theta(N^{c})$ for some $c > 0$, then the upper bound of $T$ is $C \log N$, of the same order as in the existing literature~\cite{medvedev2019continuum}.
    As shown in the simulation, the approximation error stays small over a longer period.
    Tightening of the bound will be studied in future work.
\end{rmk}

\subsection{Approximation of Equilibria}\label{sec:approx_equilibria}
Although the continuum limit captures the behavior of the finite-dimensional system only over finite time horizons, the equilibria of the latter can be approximated by those of the former when the attention parameter $u$ is near its bifurcation threshold.
Denote the leading positive unit eigenvector of $\AN$ by $\bfv^{(N)}$, and that of the graphon operator $\TW$ by $\phis$.
In addition, denote $\phiN(y) \Let \sqrt{N} \sum_{i = 1}^{N} \bfvN_{i} \II_{\BN_{i}}(y)$.
The following result states that the two functions $\phiN$ and $\phis$ are close, when the network is large.
\begin{thm}[Approximation of equilibria]\label{thm:vector}~\\\indent
    Suppose that Assumption~\ref{asmp:graphon_piecewise+positive} holds, and the graphon operator $\TW$ has a unique maximum eigenvalue $\lambda_{\max}(\TW)$ with the corresponding positive unit eigenfunction $\phis$.
    In addition, $\kappa_{N} \geq \log N / (\eps_{W} N)$.
    Then with probability $1 - N^{-r}$ for some $r > 0$, 
    \begin{align*}
        \| \phiN - \phis \|_{2} = O \bigg(\frac{1}{\sqrt{N \kappa_{N}}} \bigg).
    \end{align*}
\end{thm}
\begin{proof}
    See Appendix~\ref{sec:proof_thm:vector}.
\end{proof}

The result indicates an approximation of $\phiN$ by $\phis$ for large $N$.
That is, $\phiN(y)$ is close to $\phis(y)$ for most points $y \in \mtci$.
As a consequence, for example, the opinion mean $\int_{\mtci} \phiN(y) dy$ is similar to its limit version $\int_{\mtci} \phis(y) dy$, which will be generalized to arbitrary moments of the opinion function in the next section.

Note that from Proposition~\ref{prop_bifurcation}, the opinion vector $\bfxN$ is close to $\bfvN$ with a small error when $u$ is slightly larger than the bifurcation threshold $u^{*}$.
As a result, the opinions $\bfxN$ can be approximated by $\phis$ with an error depending on the expected degree $N \kappa_{N}$ and the difference $u - u^{*}$.

\subsection{Approximation of Opinion Distribution Functions}\label{sec:approx_distribution}
The previous approximation allows us to quantify the asymptotic behavior of the finite-dimensional system near its bifurcation.
In the following, we derive another quantification of opinion behaviors by considering its distribution function.
Consider the measure space on the unit interval $(\mtci, \mtcb(\mtci), \mu)$, where $\mtcb(\mtci)$ is the Borel $\sigma$-algebra and $\mu$ is the Lebesgue measure.
For a measurable function $g: \mtci \to \RR$, denote $F_{g}(y) \Let \mu(\{s\colon g(s) \leq y\})$, where $\mu$ is the Lebesgue measure.
Then the function $F_{g}$ is increasing, right continuous, $\lim_{y \to -\infty} F_{g}(y) = 0$, and $\lim_{y \to +\infty} F_{g}(y) = 1$.
Hence, $F_g$ is a distribution function and it defines a unique Borel measure $\mu_{g}$ on $\RR$ such that $\mu_{g}((a, b]) = F_{g}(b) - F_{g}(a)$ for all $a, b \in \RR$ ~\cite[Theorem 1.16]{folland1999real}.

For two opinion functions $\zetaN, \zetas \in L^{2}(\mtci)$, denote their distribution functions by $F_{N} \Let F_{\zetaN}$ and $F_{*} \Let F_{\zetas}$, respectively.
The difference $F_{N}(b) - F_{N}(a)$ then indicates the proportion of agents holding opinions within the interval $(a, b]$.
The distribution function $F_{N}$ is also the cumulative distribution function of $\zetaN(X)$, where $X$ is the uniform random variable on $\mtci$. 
That is, the opinion, held by a uniformly sampled agent, has probability distribution $F_{N}$.

The following result states that, when the opinion functions $\zetaN$ and $\zetas$ are close, their associated distribution functions are also close to each other.

\begin{thm}[Approximation of distribution functions]\label{thm:distribution}
    Suppose $\|\zetaN - \zetas\|_{2} \leq \eps_{N}$, where $\zetaN, \zetas \in L^{2}([0,1])$ and $\eps_{N} \geq 0$.
    Then it holds that $F_{*}(y - \sqrt{\eps_{N}}) - \eps_{N} \leq F_{N}(y) \leq F_{*}(y + \sqrt{\eps_{N}}) + \eps_{N}$.
    Furthermore, if $\eps_{N} \to 0$ as $N \to \infty$, then $\lim_{N \to \infty} F_{N}(x) = F(x)$ for all points $x \in \RR$, at which $F$ is continuous. 
\end{thm}

\begin{proof}
    See Appendix~\ref{sec:proof_thm:distribution}.
\end{proof}

The convergence of $F_{N}$ to $F_{*}$ indicates the weak convergence of the associated measures $\mu_{N}$ to $\mu_{*}$ by the Portmanteau theorem, as well as the weak convergence of the random variables $\zetaN(X)$ to $\zetas(X)$, where $X$ is the uniform random variable on $\mtci$. 
Furthermore, since $\zetaN$ and $\zetas$ are uniformly bounded, the result implies that any moment of the distribution function $F_{N}$ can be approximated by $F_{*}$, for example, the mean and variance of $\zetaN$.

From Theorem~\ref{thm:vector} and Proposition~\ref{prop_bifurcation}, we know that the final opinions $\|\thetaN(\cdot, \infty) - \theta(\cdot, \infty)\|_{2}$ can be bounded by a small error, where $\thetaN(\cdot, \infty) \Let \lim_{t \to \infty} \thetaN(\cdot, t)$ and $\theta(\cdot, \infty) \Let \lim_{t \to \infty} \theta(\cdot, t)$.
This further implies that their distribution functions are close to each other.
As a consequence, the opinion pattern in the finite-dimensional case can be quantified by examining the limit distribution over the graphon, which captures the network structure.
An illustrative example, where the graphon has a community structure shown in Fig.~\ref{fig:illus_graphs}, is given in the next section.

\begin{rmk}
    For an opinion vector $\bfx \in \RR^{N}$, the empirical measure $\mu_{\bfx}(B) \Let \frac{1}{N} \sum_{i = 1}^{N} \II_{B}(x_{i})$, defined for all Borel sets $B$, has been studied in the literature (e.g.~\cite{como2011scaling,canuto2012eulerian}).
    This is essentially the same as the measure defined by the distribution function that we introduced earlier.
    For the measure space $(\mtci, \mtcb(\mtci), \mu)$ and a measurable function $g: \mtci \to \RR$, we can define a pushforward measure of $\mu$ by $g$ as $g_{\#}\mu(B) \Let \mu(g^{-1}(B))$, where $B \in \mtcb(\mtci)$.
    Then the measure has an associated distribution function $F_{g}$ such that $F_{g}(y) = g_{\#}\mu((-\infty, y])$~\cite{folland1999real}.
\end{rmk}

\begin{figure*}[!t]
    \centering
    \subfigure[\label{fig:pde_evol}The evolution of the continuum dynamics~\eqref{eq:opinion_continuum} with initial value $\theta(y, 0) = \sin (30 \pi y)$. The nodes form two clusters in the end.]{
        \includegraphics[width=0.32\textwidth]{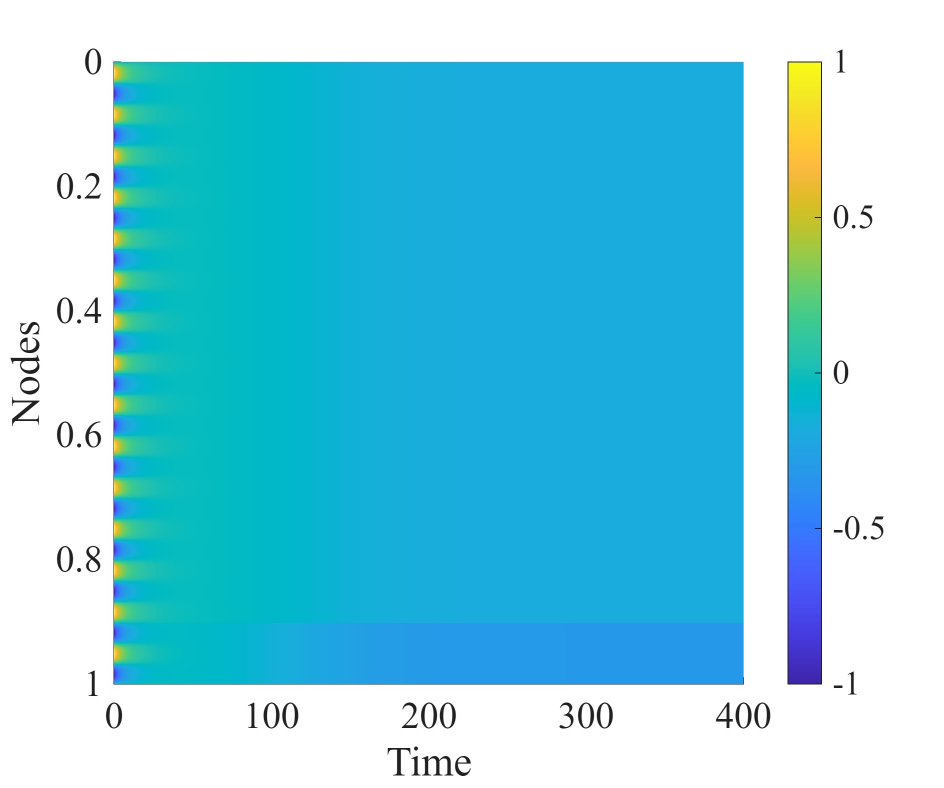}
    }  ~
    \subfigure[\label{fig:pde_sample}Evolution of $500$ nodes in the continuum dynamics, evenly sampled from the interval $\mtci$.]{
        \includegraphics[width=0.29\textwidth]{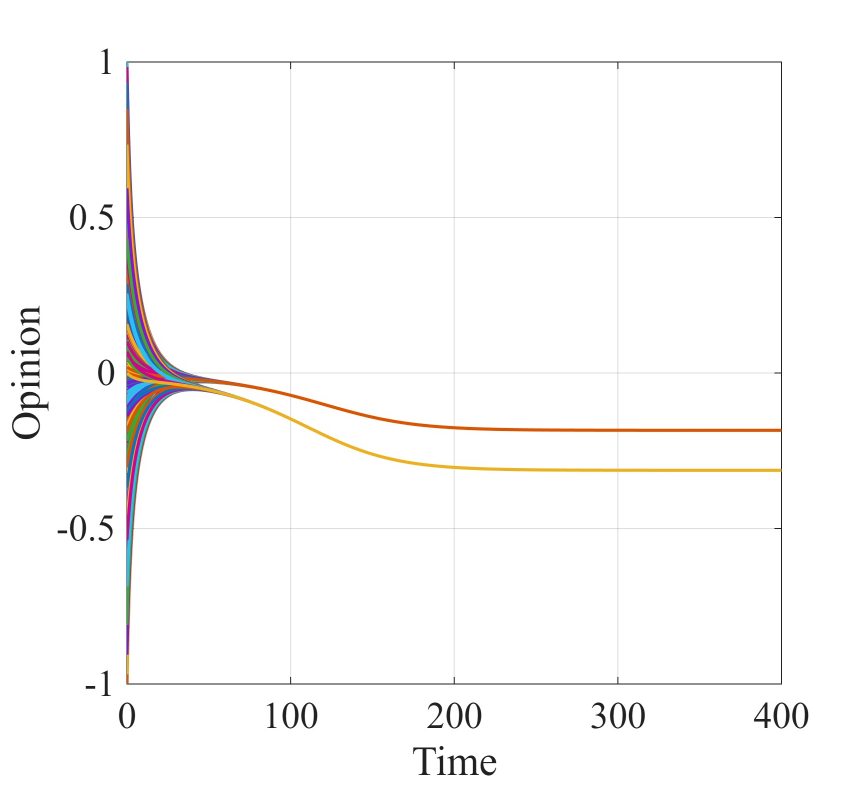}
    } ~
    \subfigure[\label{fig:ode_evol}The finite-dimensional dynamics~\eqref{eq:opinion_model_compactform} over the graph. Blue and red indicate the community labels, whereas the yellow and black dotted lines show the opinion average in each community. ]{~
        \includegraphics[width=0.29\textwidth]{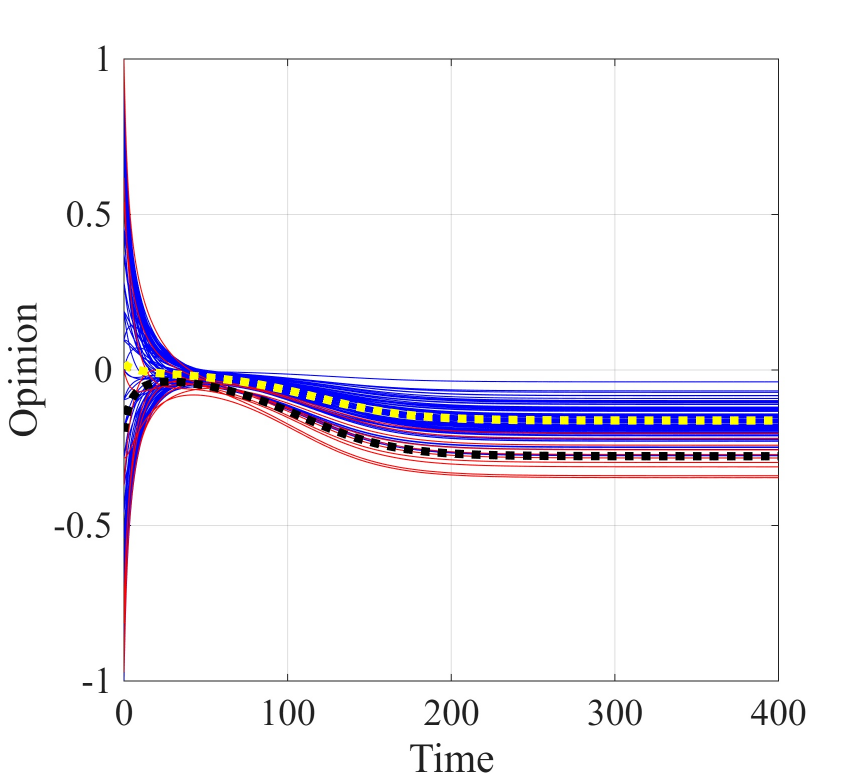}~
    }   
    \caption{\label{fig:pde_and_ode}Illustration of dynamics over a graphon given in Fig.~\ref{fig:illus_graphon} and a finite network given in Fig.~\ref{fig:illus_sbm}.}
\end{figure*}

\begin{figure*}[!t]
    \centering
        \subfigure[\label{fig:dynamic_approx_entire}Error over the entire period.]{
            ~~~
            \includegraphics[width=0.24\textwidth]{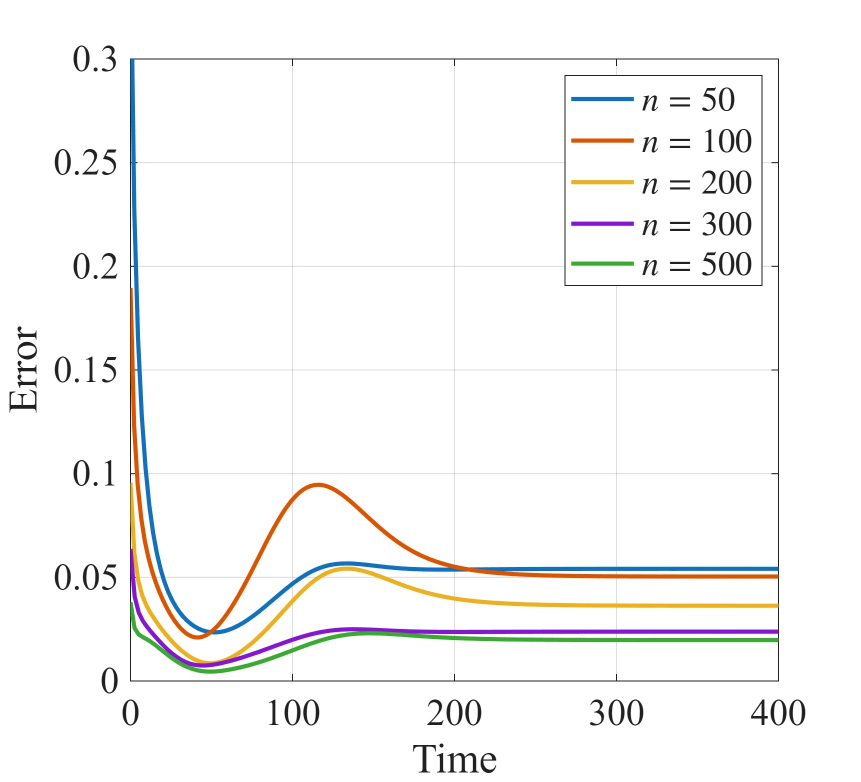} 
            ~~~
        }    
        \subfigure[\label{fig:dynamic_approx_transient}Error during the transient phase. ]{
            ~~~
            \includegraphics[width=0.24\textwidth]{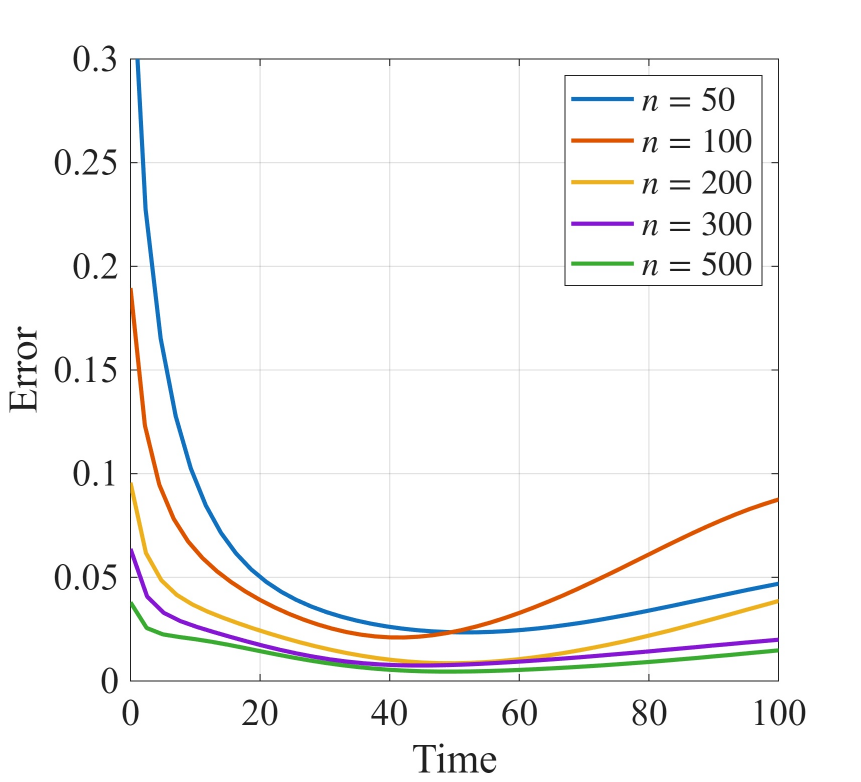}
            ~~~
        }   
        \subfigure[\label{fig:dynamic_approx_same_initial}Error during the transient phase in the case of an identical initial condition.]{
            ~~~
            \includegraphics[width=0.24\textwidth]{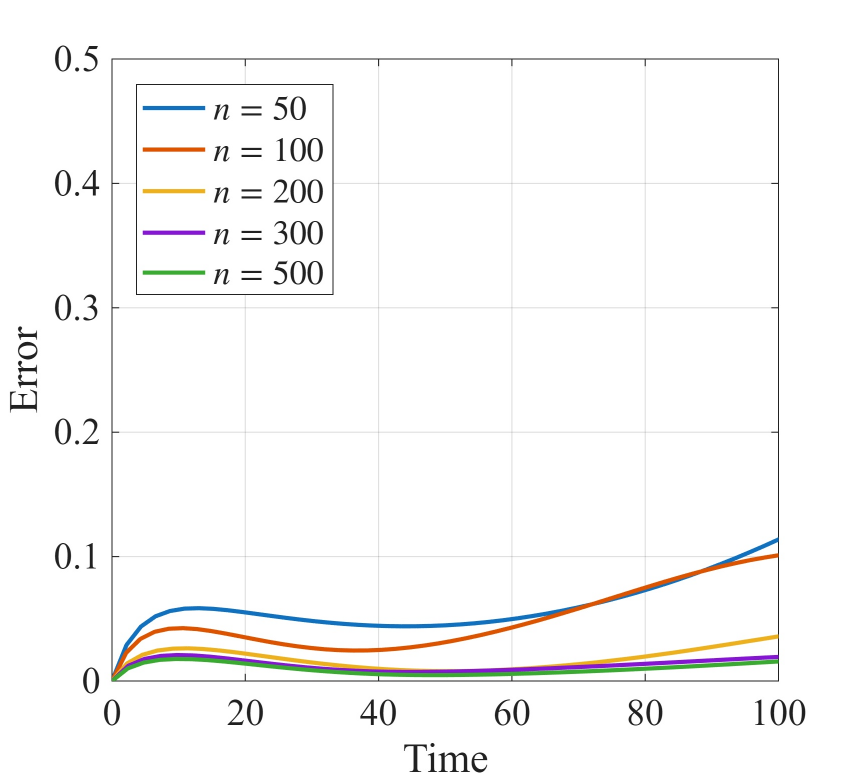}
            ~~~
        } 
    \caption{\label{fig:dynamic_approx}Error between the interpolant $\thetaN(y, t)$ and the dynamics over the graphon $\theta(y, t)$.}
\end{figure*}

\section{Simulation}\label{sec:simulation}
This section presents numerical experiments illustrating the theoretical results obtained in the previous section.

We demonstrate the results using the two-block piecewise constant graphon shown in Fig.~\ref{fig:illus_graphon}, where $W(y,z) = 0.1$ for $0 < y, z \leq 0.9$, $W(y, z) = 0.5$ for $0.9 < y, z \leq 1$, and $W(y, z) = 0.2$ otherwise, and the associated stochastic block model with finite network size $N = 200$, shown in Fig.~\ref{fig:illus_sbm}.
Both the dynamics over the finite graph and the graphon, with identical parameters $\bar{d}$, $\alpha$, $\gamma$, and $u$, eventually form two clusters (see Fig.~\ref{fig:pde_and_ode}).
The error between the interpolant and the dynamics over the graphon is presented in Fig.~\ref{fig:dynamic_approx}.
Fig.~\ref{fig:dynamic_approx_entire} illustrates that the approximation error decreases over the entire period, as the network size $N$ increases.
The decreasing error in Fig.~\ref{fig:dynamic_approx_transient} may result from that both systems first approach the origin due to the damping effect, which can also be observed in Figs.~\ref{fig:pde_sample} and~\ref{fig:ode_evol}.
When the two systems start with the same initial condition, the error is zero in the beginning, and then increases during the transient phase (Fig.~\ref{fig:dynamic_approx_same_initial}), resulting from the inaccuracy of the network model, as indicated by Theorem~\ref{thm:dynamic_approx}.

Next we examine the asymptotic opinion behavior $\thetaN(\cdot, \infty)$ and $\theta(\cdot, \infty)$, plotted in Figs.~\ref{fig:equi_finite} (with $N = 500$) and~\ref{fig:equi_infinite}.
The two opinion functions show similar patterns with a large cluster holding opinions around $0.18$ and a small one around $0.31$.
The two clusters are captured by the community structure of the graphon given in Fig.~\ref{fig:illus_graphon}.
The distribution functions of $\thetaN(\cdot, \infty)$ and $\theta(\cdot, \infty)$ are presented in Fig.~\ref{fig:equi_dist}, showing that the former becomes closer to the latter, as the network size grows.
Note that the distribution function of the final continuum opinions clearly has two jumps, corresponding to the two clusters in the network.
Fig.~\ref{fig:equi_moment} further shows that the moments of the finite-dimensional opinions converge to those of the continuum opinions, as the network size grows. 
In this example, the opinion clusters correspond to the communities, because the latter have different centrality.
Such correspondence will not hold, if the communities have equal size (see e.g.,~\cite{xing2024learning}), indicating the effect of nonlinear dynamics.

\begin{figure*}[!t]
    \centering
        \subfigure[\label{fig:equi_finite}The plot (left) and histogram (right) of $\thetaN(\cdot, \infty)$ with $N = 500$.]{\centering
        ~~~
            \includegraphics[width=0.21\textwidth]{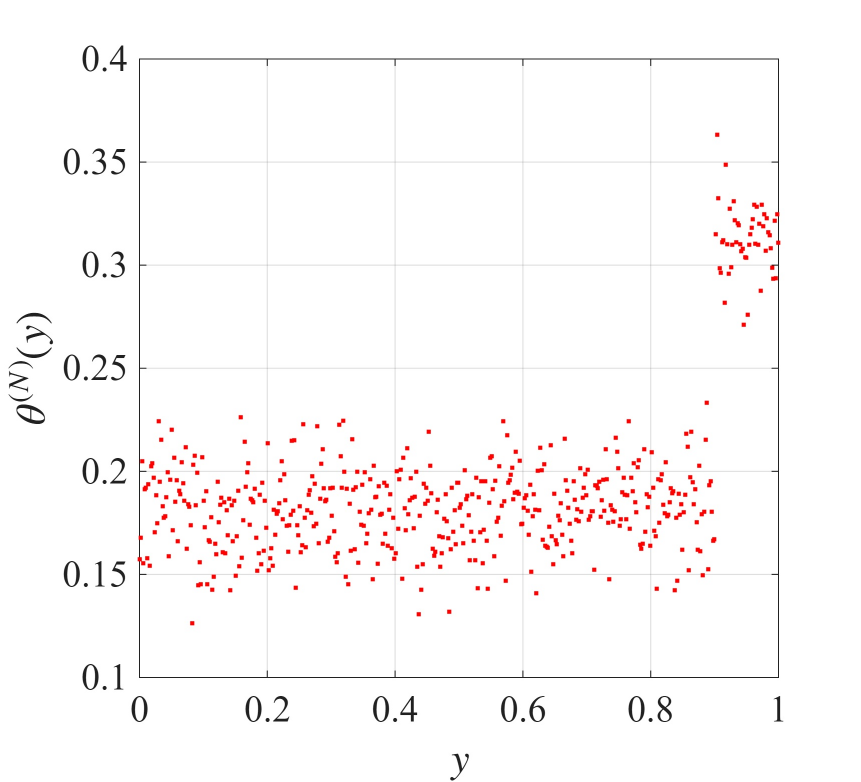} 
            \includegraphics[width=0.21\textwidth]{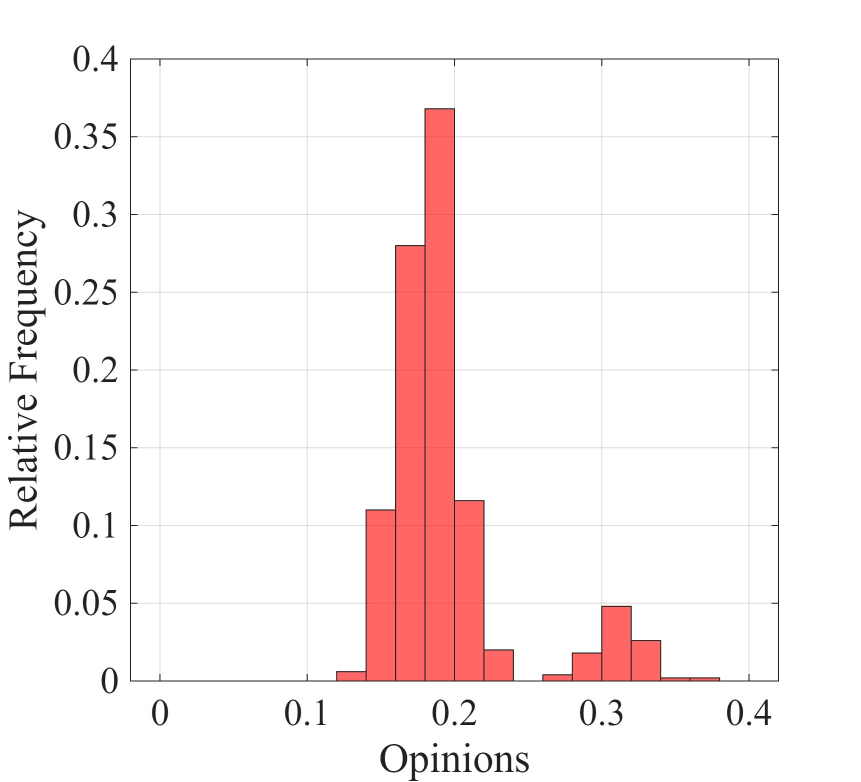} 
            ~~~
        }    
        \subfigure[\label{fig:equi_infinite}The plot (left) and histogram (right) of $\theta(\cdot, \infty)$. ]{
            ~~~
            \includegraphics[width=0.21\textwidth]{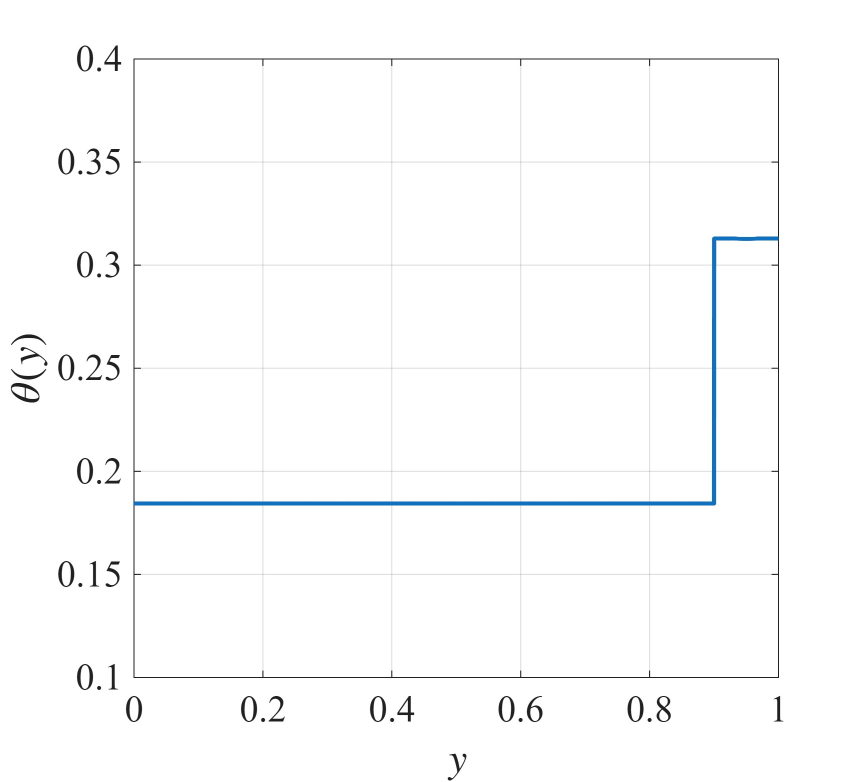}
            \includegraphics[width=0.21\textwidth]{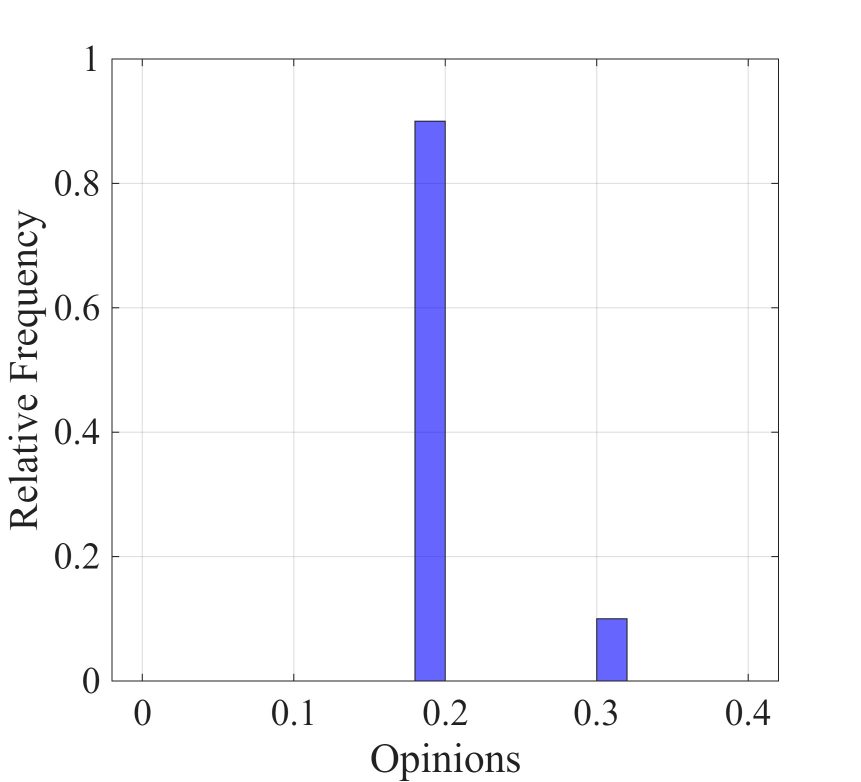}
            ~~~
        }   
        \subfigure[\label{fig:equi_dist}The distribution functions of $\thetaN(\cdot, \infty)$ and $\theta(\cdot, \infty)$ (red and blue, resp.) for $N = 100$, $500$, and $1000$.]{
            \includegraphics[width=0.22\textwidth]{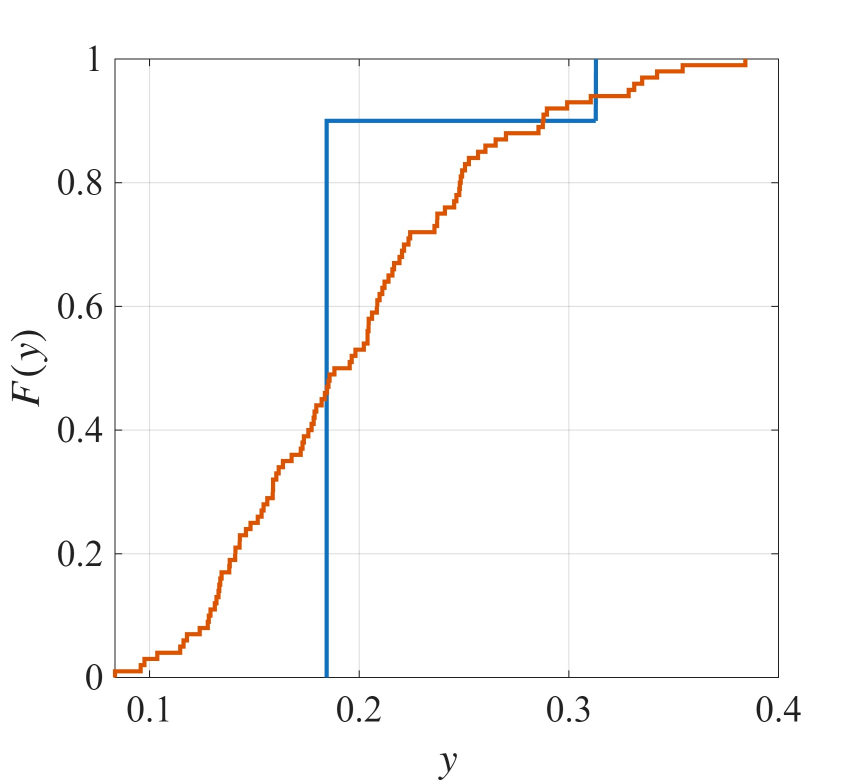}
            \includegraphics[width=0.22\textwidth]{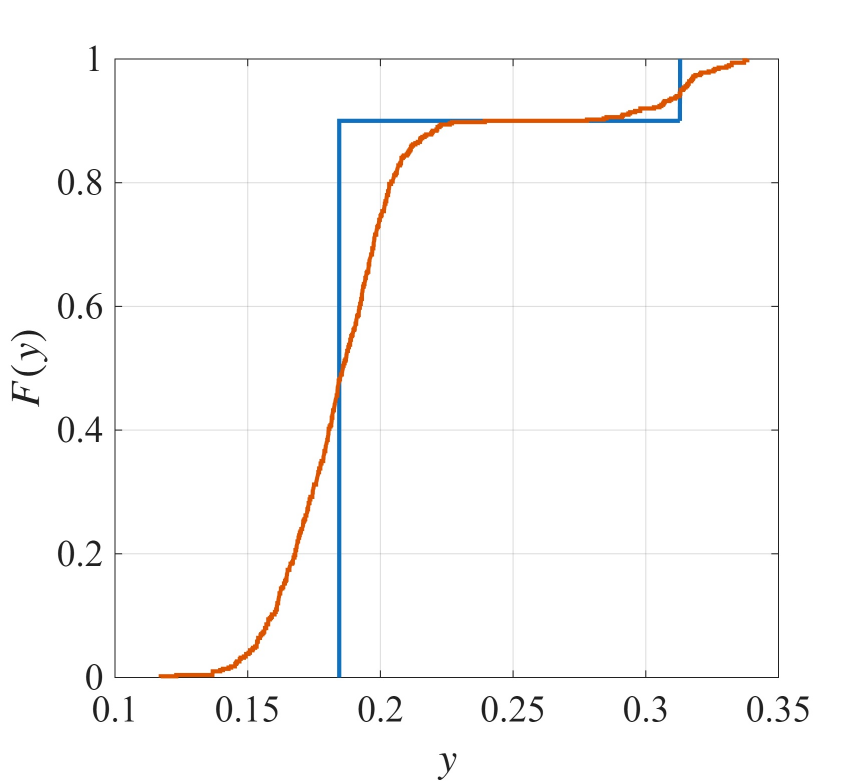}
            \includegraphics[width=0.22\textwidth]{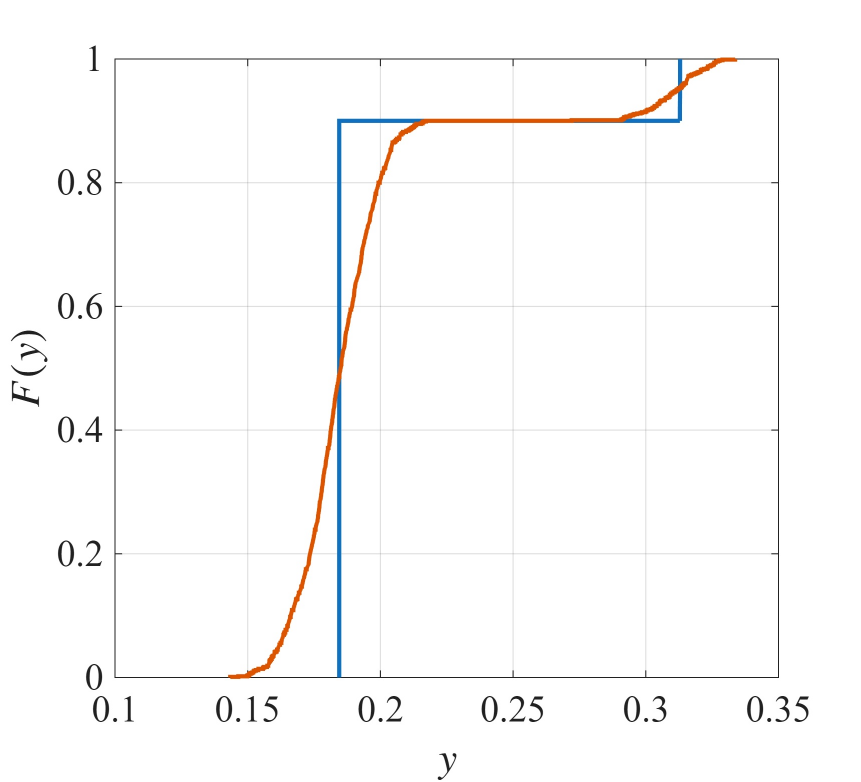}~~
        } 
        \subfigure[\label{fig:equi_moment}Error in moment approximation.]{~
            \includegraphics[width=0.22\textwidth]{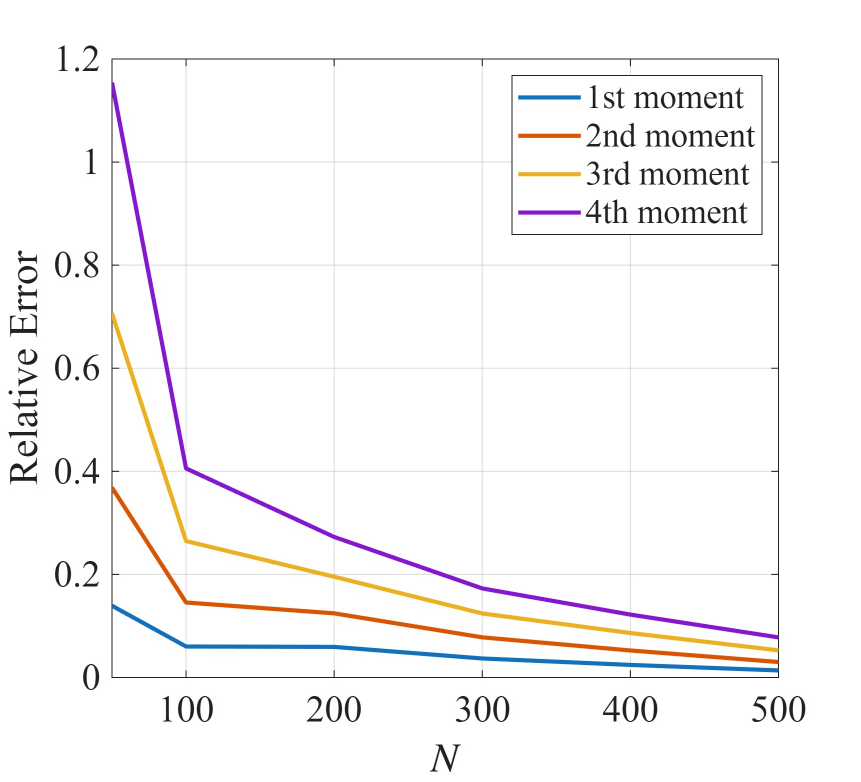}
            ~
        } 
    \caption{\label{fig:equi_approx}Approximation of the opinion $\thetaN(\cdot, \infty)$ by $\theta(\cdot, \infty)$.}
\end{figure*}

\section{Conclusions}\label{sec:conclusion}
In this paper, we studied behavior of nonlinear opinion dynamics over large-scale networks using graphons. 
We showed that both dynamics and equilibria of the original dynamics can be approximated by the dynamics over the graphon, when the network size is large. 
This result further enables a quantitative characterization of the opinion distribution by its continuum limit.
Future work includes studying multiple option case and extending the framework to data-driven modeling of opinion dynamics.




\appendix
\subsection{Proof of Theorem~\ref{thm:dynamic_approx}}\label{sec:proof_thm:dynamic_approx}
By the fundamental theorem of calculus, we have that 
\begin{align*}
    \xiN(y, t) = \xiN(y, 0) + \int_{0}^{t} \partial_{\tau} \xiN(y, \tau) d\tau, ~ y \in \mtci,
\end{align*}
where $\partial_{t} \xiN(y, t) = \partial_{t} \thetaN(y, t) - \partial_{t} \theta(y, t)$.
Then it follows from the Minkowski integral inequality that
\begin{align*}
    \|\xiN(\cdot, t)\|_{2} \leq \|\xiN(\cdot, 0)\|_{2} + \int_{0}^{t} \|\partial_{\tau} \xiN(y, \tau)\|_{2} d\tau.
\end{align*}
Note that
\begin{align*}
    &\| \partial_{t} \xiN(\cdot, t) \|_{2} \leq \bar{d} \| \xiN(\cdot, t) \|_{2} + \| \bN - b \|_{2} \\
    &\quad~ + u \Big\| S \Big(\alpha \thetaN(\cdot, t) + \gamma \int_{\mtci} \hatWN(\cdot, z) \thetaN(z, t) dz \Big) \\
    &\quad~ - S \Big(\alpha \theta(\cdot, t) + \gamma \int_{\mtci} W(\cdot, z) \theta(z, t) dz \Big) \Big\|_{2} \\
    &\leq \bar{d} \| \xiN(\cdot, t) \|_{2} + \| \bN - b \|_{2} \\
    &\quad~ + u \Big\| \alpha \thetaN(\cdot, t) + \gamma \int_{\mtci} \hatWN(\cdot, z) \thetaN(z, t) dz  \\
    &\quad~ - \Big(\alpha \theta(\cdot, t) + \gamma \int_{\mtci} W(\cdot, z) \theta(z, t) dz \Big) \Big\|_{2} \\
    &\leq (u \alpha + \bar{d}) \| \xiN(\cdot, t) \|_{2} + \| \bN - b \|_{2} + u |\gamma|  \\
    &\quad~ \Big\| \int_{\mtci} \hatWN(\cdot, z) \thetaN(z, t) dz - \int_{\mtci} W(\cdot, z) \theta(z, t) dz \Big\|_{2}, 
\end{align*}
where the second last inequality is from the Lipschitz property of $S$.
Decompose the difference below as follows
\begin{align*}
    &\Big\| \int_{\mtci} \hatWN(\cdot, z) \thetaN(z, t) dz - \int_{\mtci} W(\cdot, z) \theta(z, t) dz \Big\|_{2} \\
    &\leq \Big\| \int_{\mtci} \hatWN(\cdot, z) \thetaN(z, t) dz - \int_{\mtci} W(\cdot, z) \thetaN(z, t) dz \Big\|_{2} \\
    &\quad~ + \Big\| \int_{\mtci} W(\cdot, z) \thetaN(z, t) dz - \int_{\mtci} W(\cdot, z) \theta(z, t) dz \Big\|_{2} \\
    &= \Big\| \int_{\mtci} (\hatWN(\cdot, z) - W(\cdot, z)) \thetaN(z, t) dz \Big\|_{2} \\
    &\quad~ + \Big\| \int_{\mtci} W(\cdot, z) (\thetaN(z, t) - \theta(z, t)) dz \Big\|_{2}  \teL (I) + (II).
\end{align*}

For the first term, from the definition of an integral linear operator, we have that
\begin{align*}
    (I) &= \Big\| \int_{\mtci} (\hatWN(\cdot, z) - W(\cdot, z)) \thetaN(z, t) dz \Big\|_{2} \\
    &= \| (\ThWN - \TW) \thetaN(\cdot, t) \|_{2} \\
    &\leq \| \ThWN - \TW \|_{2} \| \thetaN(\cdot, t) \|_{2} \\
    &\leq C_{\tx{s}} \| \ThWN - \TW \|_{2},
\end{align*}
where the last inequality follows from the boundedness of the finite-dimensional system~\eqref{eq:opinion_model_compactform} and the constant $C_{\tx{s}}$ does not depend on time $t$ or the network size $n$, but only on the parameters $u$, $\alpha$, $\bar{d}$, and $\gamma$, and on the initial condition $\theta(\cdot, 0)$.

For the second term, it holds that
\begin{align*}
    (II) &= \Big\| \int_{\mtci} W(\cdot, z) (\thetaN(z, t) - \theta(z, t)) dz \Big\|_{2} \\
    &\leq \| \TW \|_{2} \| \xiN(\cdot, t) \|_{2}
    \leq \| \xiN(\cdot, t) \|_{2}.
\end{align*}

To sum up, 
\begin{align*}
    &\| \partial_{t} \xiN(\cdot, t) \|_{2} \leq (u (\alpha + |\gamma|) + \bar{d}) \| \xiN(\cdot, t) \|_{2} \\
    & + \| \bN - b \|_{2} + C_{\tx{s}} u |\gamma|  \| \ThWN - \TW \|_{2},
\end{align*}
which implies that
\begin{align*}
    &\|\xiN(\cdot, t)\|_{2} \leq \|\xiN(\cdot, 0)\|_{2} + \int_{0}^{t} (u (\alpha + |\gamma|) + \bar{d})  \\
    &\Big( \| \xiN(\cdot, t) \|_{2} + \frac{\| \bN - b \|_{2} + C_{\tx{s}} u |\gamma|  \| \ThWN - \TW \|_{2}}{u (\alpha + |\gamma|) + \bar{d}} \Big)  d\tau.
\end{align*}
From the Gronwall inequality, it follows that
\begin{multline*}
    \|\xiN(\cdot, t)\|_{2} \leq \Big( \frac{\| \bN - b \|_{2} + C_{\tx{s}} u |\gamma|  \| \ThWN - \TW \|_{2}}{u (\alpha + |\gamma|) + \bar{d}}  \Big.\\
    \Big. + \|\xiN(\cdot, 0)\|_{2} \Big) e^{(u (\alpha + |\gamma|) + \bar{d}) t}.
\end{multline*}
Since $\theta(\cdot, 0)$ and $b$ are piecewise Lipschitz, we know that both $\| \bN - b \|_{2} = O(1 / \sqrt{N})$ and $\|\xiN(\cdot, 0)\|_{2} = O(1 / \sqrt{N})$.

The bound on $\| \ThWN - \TW \|_{2}$ can be obtained using standard concentration analysis~\cite{avella2018centrality,prisant2025asymptotic}, and we provide a proof sketch here.
First, we need to bound the spectral norm of the difference between $\TW$ and $\TWN$, induced by $\GW$ and the step graphon $\WN$, respectively. 
The latter is defined by the weight matrix $\barAN$ given in Definition~\ref{defn:random_graph} as $\WN(y, z) \Let \sum_{k = 1}^{N} \sum_{\ell = 1}^{N} \baraN_{k \ell} \II_{\BN_{k}}(y) \II_{\BN_{\ell}}(z)$, $y, z \in \mtci$.
The bound is $\| \TWN - \TW \|_{2} = O(1/ \sqrt{N})$, under Assumption~\ref{asmp:graphon_piecewise+positive}~(i).
Next, the difference between the sampled graphon operator $\ThWN$ and the step graphon operator $\TWN$ can be bounded by $\| \ThWN - \TWN \|_{2} = O(1 / \sqrt{N \kappa_{N}})$, if $\kappa_{N} \geq \log N / (\eps_{W} N)$.
This follows from a bound on the adjacency matrix $\AN$ from its expectation $\kappa_{N} \barAN$, $\| \AN - \kappa_{N} \barAN \|_{2} =  O(\sqrt{N \kappa_{N}})$, implied by Theorem 5.2 of \cite{lei2015consistency} and $\| \ThWN - \TWN \|_{2} = \| \AN / \kappa_{N} - \barAN \|_{2} / N$ (see e.g., \cite[Lemmas 2 and 4]{avella2018centrality}).
The conclusion then follows.

\subsection{Proof Sketch of Theorem~\ref{thm:vector}}\label{sec:proof_thm:vector}
The proof is similar to that of Theorem~2 in~\cite{avella2018centrality}.
As in the proof of Theorem~\ref{thm:dynamic_approx}, first we need to obtain bounds $\| \TWN - \TW \|_{2} = O(1/ \sqrt{N})$ and $\| \ThWN - \TWN \|_{2} = O(1 / \sqrt{N \kappa_{N}})$.
Note that the condition $\kappa_{N} \geq \log N / (\eps_{W} N)$ ensures that the random graph is connected with high probability, and hence that the adjacency matrix has a unique largest eigenvalue.
Then the conclusion follows from the Davis--Kahan theorem (see e.g., Lemma~8 in Appendix B of~\cite{avella2018centrality}).

\subsection{Proof of Theorem~\ref{thm:distribution}}\label{sec:proof_thm:distribution}
Denote the set $\mtcs_{N} \Let \{ |\zetaN - \zetas| \geq \sqrt{\eps_{N}}\}$.
Then 
\begin{align*}
    &\int_{\mtci} |\zetaN(y) - \zetas(y)|^2 dy \geq \int_{\mtcs_{N}} |\zetaN(y) - \zetas(y)|^2 dy \\
    &\geq \int_{\mtcs_{N}} \eps_{N} dy = \eps_{N} \mu(\mtcs_{N}),
\end{align*}
implying that $\mu(\mtcs_{N}) \leq \eps_{N}$.
Since $|\zetaN(y) - \zetas(y)| \leq \sqrt{\eps_{N}}$ on $\mtci \setminus \mtcs_{N}$, it holds that $\{s\colon \zetaN(s) \leq y\} \cap (\mtci \setminus \mtcs_{N}) \subset \{s\colon \zetas(s) \leq y + \sqrt{\eps_{N}} \} \cap (\mtci \setminus \mtcs_{N})$, and
\begin{align*}
    &\mu(\{s\colon \zetaN(s) \leq y \}) = \mu(\{s\colon \zetaN(s) \leq y\} \cap (\mtci \setminus \mtcs_{N})) \\
    &\quad~ + \mu(\{s\colon \zetaN(s) \leq y\} \cap \mtcs_{N}) \\
    &\leq \mu(\{s\colon \zetas(s) \leq y + \sqrt{\eps_{N}} \} \cap (\mtci \setminus \mtcs_{N})) + \eps_{N} \\
    &\leq \mu(\{s\colon \zetas(s) \leq y + \sqrt{\eps_{N}} \}) + \eps_{N}.
\end{align*}
On the other hand, $\{s\colon \zetaN(s) \leq y\} \cap (\mtci \setminus \mtcs_{N}) \supset \{s\colon \zetas(s) \leq y - \sqrt{\eps_{N}} \} \cap (\mtci \setminus \mtcs_{N})$, implying that
\begin{align*}
    &\mu(\{s\colon \zetaN(s) \leq y \}) 
    \geq \mu(\{s\colon \zetas(s) \leq y - \sqrt{\eps_{N}} \} ) - \mu(\{s\colon  \\
    & \zetas(s) \leq y - \sqrt{\eps_{N}}\} \cap \mtcs_{N}) \geq \mu(\{s\colon \zetas(s) \leq y - \sqrt{\eps_{N}} \}) - \eps_{N}.
\end{align*}
To sum up, we have that $F_{*}(y - \sqrt{\eps_{N}}) - \eps_{N} \leq F_{N}(y) \leq F_{*}(y + \sqrt{\eps_{N}}) + \eps_{N}$.
If $\eps_{N} \to 0$ as $N \to \infty$, then $F_{N}(y) \to F_{*}(y)$ at all continuous points $y$ of $F_{*}$.

\bibliographystyle{ieeetr}
\bibliography{bibliography}

\end{document}